\documentclass[manyauthors]{fundam}

\publyear{2021}
\papernumber{0001}
\volume{178}
\issue{1}

\usepackage[T1]{fontenc}
\usepackage{graphicx}

\usepackage{graphics}
\usepackage[ruled,linesnumbered]{algorithm2e}
\usepackage{amsmath}
\usepackage{tabularx}
\usepackage{environ}
\usepackage{amssymb}
\usepackage{todonotes}

\usepackage{thmtools} 
\usepackage{thm-restate}

\usepackage{mathtools}
\usepackage[super]{nth}
\usepackage{todonotes}
\usepackage{hyperref}
\usepackage{color}

\usepackage[capitalise,nameinlink]{cleveref}
\crefname{problem}{Problem}{Problems}
\crefname{Theorem}{Theorem}{Theorems}
\crefname{Claim}{Claim}{Claims}
\crefname{Remark}{Remark}{Remarks}

\makeatletter
\newcommand{\problemtitle}[1]{\gdef\@problemtitle{#1}}
\newcommand{\probleminput}[1]{\gdef\@probleminput{#1}}
\newcommand{\problemquestion}[1]{\gdef\@problemquestion{#1}}
\NewEnviron{problemdescription}{
  \problemtitle{}\probleminput{}\problemquestion{}
  \BODY
  \par\addvspace{.5\baselineskip}
  \noindent
  \begin{tabularx}{\textwidth}{@{\hspace{\parindent}} l X c}
    \hline
    \multicolumn{2}{@{\hspace{\parindent}}l}{\@problemtitle} \\
    \textbf{Input:} & \@probleminput \\
    \textbf{Question:} & \@problemquestion\\
    \hline
  \end{tabularx}
  \par\addvspace{.5\baselineskip}
}
\makeatother

\DeclareMathOperator{\al}{alph}

\DeclareMathOperator{\ETH}{\textsf{ETH}}

\DeclareMathOperator{\OV}{\textsf{OV}}
\DeclareMathOperator{\OVH}{\textsf{OVH}}
\DeclareMathOperator{\pmas}{p-MAS}
\DeclareMathOperator{\poly}{poly}

\DeclareMathOperator{\psas}{p-SAS}
\DeclareMathOperator{\nextpos}{next}

\DeclareMathOperator{\sas}{SAS}

\DeclareMathOperator{\SatProb}{\textsc{CNF-Sat}}
\DeclareMathOperator{\SETH}{\textsf{SETH}}
\DeclareMathOperator{\smallO}{o}

\DeclareMathOperator{\Subseq}{Subseq}

\newcommand{\pSubSeqMatch}{\mathtt{pSubSeqMatch}}
\newcommand{\SubSeqMatch}{\mathtt{SubSeqMatch}}
\newcommand{\ItSubSeqMatch}{\mathtt{ItSubSeqMatch}}
\newcommand{\BestItSubSeqMatch}{\mathtt{BestItSubSeqMatch}}

\newcommand{\kpNonUniv}{\mathtt{kpNonUniv}}
\newcommand{\kpNonEquiv}{\mathtt{kpNonEquiv}}
\newcommand{\pWords}{\mathtt{partialWordsNonUniv}}
\newcommand{\pNonPSAS}{\mathtt{pNonPSAS}}
\newcommand{\pMAS}{\mathtt{pMAS}}

\newcommand{\xSubseq}[1]{#1$-$\Subseq}
\newcommand{\pSubseq}{\xSubseq{p}}
\def\nth#1{#1$^{\text{th}}$}
\newcommand{\len}[1]{|#1|}
\DeclareMathAlphabet{\mathbbb}{U}{bbold}{m}{n}

\begin{document}

\title{Subsequences in Bounded Ranges: Matching and Analysis Problems}
\author{Maria Kosche\\ Georg-August University Göttingen \and
	Tore Koß\corresponding\\ Georg-August University Göttingen \and
	Florin Manea\\ Georg-August University Göttingen \and
	Viktoriya Pak\\ Georg-August University Göttingen}
\runninghead{M. Kosche, T. Koß, F. Manea, V. Pak}{Subsequences in Bounded Ranges}
%
	\address{\{maria.kosche,tore.koss,florin.manea,viktoriya.pak\}@cs.uni-goettingen.de}
	\maketitle              
\begin{abstract}
In this paper, we consider a variant of the classical algorithmic problem of checking whether a given word $v$ is a subsequence of another word $w$. More precisely, we consider the problem of deciding, given a number $p$ (defining a range-bound) and two words $v$ and $w$, whether there exists a factor $w[i:i+p-1]$ (or, in other words, a range of length $p$) of $w$ having $v$ as subsequence
(i.\,e., $v$ occurs as a subsequence in the bounded range $w[i:i+p-1]$). We give matching upper and lower quadratic bounds for the time complexity of this problem. Further, we consider a series of algorithmic problems in this setting, in which, for given integers $k$, $p$ and a word $w$, we analyse the set $\pSubseq_{k}(w)$ of all words of length $k$ which occur as subsequence of some factor of length $p$ of $w$. Among these, we consider the $k$-universality problem, the $k$-equivalence problem, as well as problems related to absent subsequences. Surprisingly, unlike the case of the classical model of subsequences in words where such problems have efficient solutions in general, we show that most of these problems become intractable in the new setting when subsequences in bounded ranges are considered. Finally, we provide an example of how some of our results can be applied to subsequence matching problems for circular words.
	
	\keywords{Subsequences \and Bounded Range \and Matching Problems \and Analysis Problems \and Algorithms \and Fine Grained Complexity}
\end{abstract}

\section{Introduction}

A word $u$ is a subsequence of a string $w$ if there exist (possibly empty) strings $v_1, \ldots, v_{\ell+1}$ and $u_1, \ldots, u_\ell$ such that $u = u_1 \ldots u_\ell$ and $w = v_1 u_1 \ldots v_\ell u_\ell v_{\ell+1}$. In other words, $u$ can be obtained from $w$ by removing some of its letters. In this paper, we focus on words occurring as subsequences in bounded ranges of a word. \looseness=-1

The notion of subsequences appears in various areas of computer science. For instance, in automata theory, the theory of formal languages, and logics, it is used in connection with piecewise testable languages~\cite{simonPhD,Simon72,KarandikarKS15,CSLKarandikarS,journals/lmcs/KarandikarS19}, or subword order and downward closures~\cite{HalfonSZ17,KuskeZ19,Kuske20,Zetzsche16,BaumannGTZ22}. Naturally, subsequences appear in the area of combinatorics and algorithms on words~\cite{RigoS15,FreydenbergerGK15,LeroyRS17a,Rigo19,Seki12,Mat04,Salomaa05,Barker2020,kiel2022}, but they are also used for modelling concurrency~\cite{Riddle1979a,Shaw1978,BussSoltys2014}, as well as in database theory (especially in connection with \emph{event stream processing}~\cite{ArtikisEtAl2017,GiatrakosEtAl2020,ZhangEtAl2014}). Nevertheless, a series of classical, well-studied, and well-motivated combinatorial and algorithmic problems deal with subsequences. Some are stringology problems, such as {the longest common subsequence problem} \cite{DBLP:journals/tcs/Baeza-Yates91,DBLP:conf/fsttcs/BringmannC18,BringmannK18,AbboudEtAl2015,AbboudEtAl2014} or {the shortest common supersequence problem} \cite{Maier:1978}), but there are also problems related to the study of patterns in permutations, such as increasing subsequences or generalizations of this concept \cite{Erdos1935,Dilworth1950,BESPAMYATNIKH20007,Crochemore2010,Hunt,Fredman75,BiedlBCLMNS19,GawrychowskiMS22}. \looseness=-1

In general, one can split the algorithmic questions related to the study of subsequences in two large classes. The first class of problems is related to \emph{matching} (or searching), where one is interested in deciding whether a given word $u$ occurs as a subsequence in another (longer) word $w$. The second class contains \emph{analysis problems}, which are focused on the investigation of the sets $\Subseq_k(w)$ of all subsequences of length $k$ of a given string $w$ (of course, we can also remove the length restriction, and investigate the class of all subsequences of $w$). In this setting, one is interested, among other problems, in deciding the {\em $k$-universality problem}, i.\,e., whether $\Subseq_k(w)=\Sigma^k$, where $w$ and $k$ are given, or the {\em equivalence problem}, i.\,e., whether $\Subseq_k(w)=\Subseq_k(u)$, where $w,u$ and $k$ are given. In the general case of subsequences, introduced above, the matching problem can be solved trivially by a greedy algorithm. The case of analysis problems is more interesting, but still well-understood (see, e.\,g., \cite{Pin2004,Pin2019}). For instance, the equivalence problem, which was introduced by Imre Simon in his PhD thesis \cite{simonPhD}, was intensely studied in the combinatorial pattern matching community (see \cite{TCS::Hebrard1991,garelCPM,SimonWords,DBLP:conf/wia/Tronicek02,CrochemoreMT03,KufMFCS} and the references therein). This problem was optimally solved in 2021 \cite{GawrychowskiEtAl2021}. The universality problem was also intensely studied (see \cite{Barker2020,DayFKKMS21} and the references therein); to this end, we will also recall the work on absent subsequences in words \cite{Kosche2021}, where the focus is on minimal strings (w.\,r.\,t. length or the subsequence relation) which are not contained in $\Subseq(u)$. \looseness=-1

Getting closer to the topic of this paper, let us recall the following two scenarios related to the motivation for the study of subsequences, potentially relevant in the context of reachability and avoidability problems. Assume that $w$ is some string (or stream) we observe, which may represent, on the one hand, the trace of some computation or, on the other hand, and in a totally different framework, the DNA-sequence describing some gene. Deciding whether a word $v$ is a subsequence of $w$ can be interpreted as deciding, in the first case, whether the events described by the trace $v$ occurred during the longer computation described by $w$ in the same order as in $v$, or, in the second case, if there is an alignment between the sequence of nucleotides $v$ and the longer sequence $w$. However, in both scenarios described above, it seems unrealistic to consider occurrences of $v$ in $w$ where the positions of $w$ matching the first and last symbol of $v$, respectively, are very far away from each other. It seems indeed questionable, for instance, whether considering an alignment of DNA-sequences $v$ and $w$ where the nucleotides of $v$ are spread over a factor of $w$ which is several times longer than $v$ itself is still meaningful. Similarly, when observing a computation, one might be more interested in its recent history (and the sequences of events occurring there), rather than analysing the entire computation. Moreover, the fact that in many practical scenarios (including those mentioned above) one has to process streams, which, at any moment, can only be partly accessed by our algorithms, enforces even more the idea that the case where one is interested in subsequences occurring arbitrarily in a given string is less realistic and less useful than the case where one is interested in the subsequences occurring in bounded ranges of the respective string (which can be entirely accessed and processed at any moment by our algorithms).\looseness=-1

Hence, we consider in this paper the notion of $p$-subsequence of a word. More precisely, a word $v$ is a $p$-subsequence of $w$ if there exists a factor (or bounded range) $w[i:i + p - 1]$ of $w$, of length $p$, such that $v$ is a subsequence of $w[i:i + p - 1]$. In this framework, we investigate both matching and analysis problems. \looseness=-1

{\em Our results.} With respect to the matching problem, we show that checking, for a given integer $p$, whether a word $v$ is a $p$-subsequence of another word $w$ can be done in $\mathcal{O}(|w||v|)$ time and $\mathcal{O}(|v|)$ space, and show that this is optimal conditional to the Orthogonal Vectors Hypothesis. With respect to the analysis problem, we show that the problem of checking, for given integers $k$ and $p$ and word $w$, whether there exists a word $v$ of length $k$ which is not a $p$-subsequence of $w$ is NP-hard. Similarly,  checking, for given integer $p$ and words $v$ and $w$, whether the sets of $p$-subsequences of $w$ and $v$ are not equal is also NP-hard. These results are complemented by conditional lower bounds for the time complexity of algorithms solving them. Several results related to the computation of absent $p$-subsequences are also shown. Interestingly, we show that checking if a word is a shortest absent $p$-subsequence of another word is NP-hard, while checking if a word is a minimal absent $p$-subsequence of another word can be done in quadratic time (and this is optimal conditional to the Orthogonal Vectors Hypothesis). We end the paper with a series of results related to subsequences in circular words. Among other results, we also close an open problem from \cite{NagyProperties} that is related to the computation of minimal representatives of circular words. \looseness=-1

{\em Related works.} Clearly, considering properties of bounded ranges (or factors of bounded length) in words can be easily related to the study of sliding window algorithms \cite{GanardiHKLM18,GanardiHL16,GanardiHL18,GanardiHL18b,GanardiHLS19} or with algorithms in the streaming model \cite{BathieS21,DudekGGS22}. For our algorithmic results, we discuss their relation to results obtained in those frameworks. Moreover, the notion of subsequences with gap constraints was recently introduced and investigated \cite{Day2022}. In that case, one restricts the occurrences of a subsequence $v$ in a word $w$, by placing length (and regular) constraints on the factors of $w$ which are allowed to occur between consecutive symbols of the subsequence when matched in $w$. Our framework (subsequences in bounded range) can be, thus, seen as having a general length constraint on the distance between the position of $w$ matching the first symbol of $v$ and the position of $w$ matching the last symbol of $v$, in the occurrences of $v$ inside $w$. \looseness=-1

{\em Structure of the paper.} We first give a series of preliminary definitions. Then we discuss the matching problem. Further, we discuss the analysis problems and the problems related to absent subsequences. We then discuss the case of subsequences of circular words. We end with a conclusions section. \looseness=-1


\section{Basic Definitions}\label{sec:defs}

Let $\mathbb{N}$ be the set of natural numbers, including $0$. For $m, n \in \mathbb{N}$, we define the range (or interval) of natural numbers lower bounded by $m$ and upper bounded by $n$ as $[m:n] = \{m, m+1, \ldots, n\}$. An alphabet $\Sigma$ is a non-empty finite set of symbols (called letters). A {\em string (or word)} is a finite sequence of letters from $\Sigma$, thus an element of the free monoid $\Sigma^\ast$. Let $\Sigma^+ = \Sigma^\ast \setminus \{\varepsilon\}$, where $\varepsilon$ is the empty string. The {\em length} of a string $w \in \Sigma^\ast$ is denoted by $\len{w}$. The \nth{$i$} letter of $w \in \Sigma^\ast$ is denoted by  $w[i]$, for $i \in [1:\len w]$. Let $\len{w}_a = |\{i \in [1:\len w] \mid w[i] = a \}|$; let $\al(w) = \{x \in \Sigma \mid \len w_x > 0 \}$ be the smallest subset $S \subseteq \Sigma$ such that $w \in S^\ast$. For $m, n \in \mathbb{N}$, with $m\leq n$, we define the range (or factor) of $w$ between positions $m$ and $n$ as $w[m:n] = w[m] w[m+1] \ldots w[n]$. 
If $m>n$, then $w[m:n]$ is the empty word. Also, by convention, if $m<1$, then $w[m:n]=w[1:n]$, and if $n>|w|$, then $w[m:n]=w[m:|w|]$. A factor $u=w[m:n]$ of $w$ is called a  {\em prefix} (respectively, {\em suffix}) of $w$ if $m=1$ (respectively, $n=|w|$). \looseness=-1

The powers of a word $w$ are defined as: $w^0=\varepsilon$ and $w^{k+1}=ww^k$, for $k\geq 0$. We define $w^\omega$ as the right infinite word which has $w^n$ as prefix for all $n\geq 0$. The positive integer $p\leq |w|$ is a period of a word $w$ if $w$ is a prefix of $w[1:p]^\omega$.
Let $w = w_1 \ldots w_n$
(for some $w_1, \ldots , w_n \in \Sigma$) and $p \in \mathbb{N}$.
Then $w^{\frac{p}{n}} = w^{\lfloor{\frac{p}{n}}\rfloor} w'$,
where $w' = w_1 \ldots w_{(p \bmod n)}$.

We recall the notion of subsequence.
\begin{definition}
A word $v$ is a subsequence of length $k$ of $w$ (denoted $v\leq w$), where $\len{w} = n$, if there exist positions $1 \leq i_1 < i_2 < \ldots < i_k \leq n$, such that $v = w[i_1] w[i_2] \cdots w[i_k]$. The set of all subsequences of $w$ is denoted by $\Subseq(w)$. 
\end{definition}

When $v$ is not a subsequence of $w$ we also say that $v$ is {\em absent} from $w$. 
The main concept discussed here is that of $p$-subsequence, introduced next. 
\begin{definition}
1. A word $v$ is a $p$-subsequence of $w$ (denoted $v \leq_p w$) if there exists $i \leq \len{w} - p + 1$ such that $v$ is a subsequence of $w[i:i + p - 1]$.\\
2. For $p\in\mathbb N$ and $w\in\Sigma^*$, we denote the set of all $p$-subsequences of $w$ by $\pSubseq(w) = \{v\in \Sigma^*\mid v \leq_p w\}$. Furthermore, for $k\in \mathbb N$, we denote the set of all $p$-subsequences of length $k$ of $w$ by $\pSubseq_{k}(w)$.
\end{definition}

Extending the notions of absent subsequences introduced in \cite{Kosche2021}, we now define the notion of (shortest and minimal) absent $p$-subsequences in a word.

\begin{definition}
The word $v$ is an absent $p$-subsequence of $w$ if $v \notin \pSubseq(w)$. We also say $v$ is $p$-absent from $w$.
The word $v$ is a $\psas$ (shortest absent $p$-subsequence) of $w$ if $v$ is an absent $p$-subsequence of $w$ of minimal length.
The word $v$ is a $\pmas$ (minimal absent $p$-subsequence) of $w$ if $v$ is an absent $p$-subsequence of $w$ but all subsequences of $v$ are $p$-subsequences of $w$.
\end{definition}
Note that, in general, any shortest absent $|w|$-subsequence of a word $w$ (or, simply, shortest absent subsequence of $w$, denoted  $\sas$) has length $\iota(w) + 1$, where $\iota(w)=\max\{k\mid \al(w)^k\subseteq \Subseq(w)\}$ is the universality index of $w$ \cite{Barker2020}.

\smallskip

\noindent\emph{Computational Model.} In general, the problems we discuss here are of algorithmic nature. The computational model we use to describe our algorithms is the standard unit-cost RAM with logarithmic word size: for an input of size $N$, each memory word can hold $\log N$ bits. Arithmetic and bitwise operations with numbers in $[1:N]$ are, thus, assumed to take $\mathcal{O}(1)$ time.
  
In all the problems, we assume that we are given a word $w$ or two words $w$ and $u$, with $|w|=n$ and $|v|=m$ (so the size of the input is $N=n+m$), over an alphabet $\Sigma=\{1,2,\ldots,\sigma\}$, with $2\leq |\Sigma|=\sigma\leq n+m$. That is, we assume that the processed words are sequences of integers (called letters or symbols), each fitting in $\mathcal{O}(1)$ memory words. This is a common assumption in string algorithms: the input alphabet is said to be {\em an integer alphabet}. For more details see, e.\,g.,~\cite{crochemore}. 

Our algorithmic results (upper bounds) are complemented by a series of lower bounds. In those cases, we show that our results hold already for the case of constant alphabets. That is, they hold already when the input of the problem is restricted to words over an alphabet $\Sigma=\{1,2,\ldots,\sigma\}$, with $\sigma\in \mathcal{O}(1)$. 

\smallskip

\noindent\emph{Complexity Hypotheses.} 
As mentioned, we are going to show a series of conditional lower bounds for the time complexity of the considered problems. Thus, we now recall some standard computational problems and complexity hypotheses regarding them, respectively, on which we base our proofs of lower bounds.

The \emph{Satisfiability problem for formulas in conjunctive normal form}, in short $\SatProb$, gets as input a Boolean formula $F$ in conjunctive normal form as a set of clauses $F = \{c_1, c_2, \ldots, c_m\}$ over a set of variables $V = \{v_1, v_2, \ldots, v_n\}$, i.\,e., for every $i \in [m]$, we have $c_i \subseteq \{v_1, \neg v_1, \ldots, v_n, \neg v_n\}$. The question is whether $F$ is satisfiable. 
By $k$-$\SatProb$, we denote the variant where $|c_i| \leq k$ for all $i \in [m]$.\looseness=-1

The \emph{Orthogonal Vectors problem} ($\OV$ for short) gets as input two sets $A, B$ each containing $n$ Boolean-vectors of dimension~$d$, where $d\in \omega(\log n)$. The question is whether there exist two vectors $\vec{a} \in A$ and $\vec{b} \in B$ which are orthogonal, i.\,e., $\vec{a}[i] \cdot \vec{b}[i] = 0$ for every $i \in [1:d]$.

We shall use the following algorithmic hypotheses based on $\SatProb$ and $\OV$
that are common for obtaining conditional lower bounds in fine-grained complexity.
In the following, $\poly$ is any fixed polynomial function: 

\smallskip
\noindent -- \emph{Exponential Time Hypothesis}  ($\ETH$)~\cite{ImpagliazzoEtAl2001,LokshtanovEtAl2011}: $3$-$\SatProb$ cannot be solved in time $2^{\smallO(n)} \poly(n + m)$.
	
\noindent -- \emph{Strong Exponential Time Hypothesis} ($\SETH$)~\cite{ImpagliazzoPaturi2001,Williams2015}: For every $\epsilon > 0$ there exists $k$ such that $k$-$\SatProb$ cannot be decided in $\mathcal{O}(2^{n(1-\epsilon)} \poly(n))$ time. 

The following result, which essentially formulates the Orthogonal Vectors Hypothesis ($\OVH$), can be shown (see~\cite{Bringmann2014,Bringmann2019,Williams2015}).
\begin{lemma}
$\OV$ cannot be solved in $\mathcal{O}(n^{2-\epsilon} poly(d))$ time for any $\epsilon > 0$, unless $\SETH$ fails. 
\end{lemma}\ \\

\section{Matching Problems}

We first consider the following problem.
\begin{problemdescription}\label{p:ssqmatch}
  \problemtitle{Subsequence Matching in Bounded Range, $\pSubSeqMatch$}
  \probleminput{Two words $u$ and $w$ over $\Sigma$ and an integer $p$, with $\len{u} = m$ and $\len{w} = n$, and $m\leq p\leq n$.}
  \problemquestion{Decide whether $u \leq_p w$.}
\end{problemdescription}


\begin{theorem}\label[Theorem]{thm:algoMatching}
$\pSubSeqMatch$ can be solved in $\mathcal{O}(mn)$ time.
\end{theorem}
\begin{proof}
We present an algorithm which solves $\pSubSeqMatch$ in $\mathcal{O}(mn)$ time and works in a streaming fashion w.\,r.\,t. the word $w$. More precisely, our algorithm scans the letters of $w$ one by one, left to right (i.\,e., in the order $w[1], w[2], \ldots, w[n]$), and after scanning the letter $w[t]$, for $t\geq p$, it decides whether $u$ is a subsequence of the bounded range $w[t-p+1:t]$ (i.\,e., our algorithm works as a sliding window algorithm, with fixed window size, and the result for the currently considered algorithms is always calculated before the next letter is read). 

Let us explain how this algorithm works. We maintain an array $A[\cdot]$ with $m$ elements such that the following invariant holds. For $t\geq 0$: after the $t^{th}$ letter of $w$ is scanned and $A$ is updated, $A[i]$ is the length of the shortest suffix of $w[t-p+1:t]$ which contains $u[1:i]$ as a subsequence (or $A[i]=+\infty$ if $w[t-p+1:t]$ does not contain $u[1:i]$ as a subsequence). Note that, before reading the letter $w[t]$, for all $t\in [1:n]$, the array $A$, if correctly computed, has the property that $A[i]\leq A[i+1]$, for all $i\in [1:m]$. 

Initially, we set $A[i]=+\infty$ for all $i\in [1:m]$. 
Let us see how the elements of $A$ are updated when $w[t]$ is read. We first compute $\ell\gets |u|_{w[t]}$ and the positions $j_1, \ldots, j_\ell \in [1:m]$ such that $u[j_h]=w[t]$. Then, we compute an auxiliary array $B[\cdot]$ with $m$ elements, in which we set $B[j_h]\gets A[j_{h}-1]+1$, for $h\in [1:\ell]$. For $i\in [1:m] \setminus \{j_1,\ldots,j_\ell\}$ we set $B[i]\gets A[i]+1$. Intuitively, $B[i]$ is the length of the shortest suffix of $w[t-p:t]$ which contains $u[1:i]$ as a subsequence. Then, we update $A[i]$, for $i\in [1:m]$, by setting $A[i]\gets B[i]$ if $B[i]\leq p$ and $A[i]\gets +\infty$, otherwise.

After performing the update of the array $A[\cdot]$ corresponding to the scanning of letter $w[t]$, we decide that $u$ occurs in the window $w[t-p+1:t]$ if (and only if) $A[m]\leq p$. 

Let us show that the algorithm is correct. The stated invariant clearly holds before scanning the letters of $w$ (i.\,e., after $0$ letters were scanned). Assume that the invariant holds after $f-1$ letters of $w$ were scanned. Now, we will show that it holds after $f$ letters were scanned. Assume that, for some prefix $u[1:i-1]$ of $u$, we have that $w[f':f-1]$ is the shortest suffix of $w[f-p:f-1]$ which contains $u[1:i-1]$ as subsequence. Next, we scan letter $w[f]$. If $w[f]=u[i]$, then $w[f':f]$ is the shortest suffix of $w[f-p:f]$ which contains $u[1:i]$ as subsequence (otherwise, there would exist a suffix of $w[f-p:f-1]$ that is shorter than $w[f':f-1]$ which contains $u[1:i-1]$ as subsequence). So, it is correct to set $B[i]\gets A[i-1]+1$, when $w[f]=u[i]$. Otherwise, if $w[f]\neq u[i]$, we note that the shortest suffix of $w[f-p:f]$ which contains $u[1:i]$ as subsequence starts on the same position as the shortest suffix of $w[f-p:f-1]$ which contains $u[1:i]$ as subsequence. Therefore, $B[i]\gets A[i]+1$ is also correct (as we compute the length of these shortest suffixes w.\,r.\,t. the currently scanned position of $w$). Then, we simply update $A$  to only keep track of those suffixes of the currently considered range of size $p$, i.\,e., $w[f-p+1:f]$. 

The algorithm runs in $\mathcal{O}(nm)$ time. Indeed, for each scanned letter $w[f]$ we perform $\mathcal{O}(m)$ operations. Moreover, the space complexity of the algorithm is $\mathcal{O}(m)$, as we only maintain the arrays $A$ and $B$. So, the statement holds.
\end{proof}

As stated in the proof of \cref{thm:algoMatching}, the algorithm we presented can be seen as an algorithm in the sliding window model with window of fixed size $p$ (see \cite{GanardiHL16,GanardiHKLM18,GanardiHLS19}). More precisely, we scan the stream $w$ left to right and, when the $t^{th}$ letter of the stream is scanned, we report whether the window $w[t-p+1:p]$ contains $u$ as a subsequence. In other words, we report whether the string $w[t-p+1:p]$ is in the regular language $L_u=\{v\mid u\leq v\}$. The problem of checking whether the factors of a stream scanned by a sliding window are in a regular language was heavily investigated, see \cite{Ganardi19} and the references therein. In particular, from the results of \cite{GanardiHL16} it follows that, for a constant $u$ (i.\,e., $u$ is not part of the input), the problem of checking whether the factors of a stream scanned by a sliding window are in the language $L_u$ cannot be solved using $o(\log p)$ bits when the window size is not changing and equals $p$. We note that our algorithm is optimal from this point of view: if $u$ is constant and, thus, $m\in \mathcal{O}(1)$, our algorithm uses $\mathcal{O}(\log p)$ bits to store the arrays $A$ and $B$. 

We can show that our algorithm is also optimal (conditional to $\OVH$) also from the point of view of time complexity.
\begin{theorem}\label[Theorem]{thm:lowerBoundSubseq}
$\pSubSeqMatch$ cannot be solved in time $\mathcal{O}(n^h m^g)$, where $h+g= 2-\epsilon$ with $\epsilon>0$, unless $\OVH$ fails. 
\end{theorem}
\begin{proof}

Let $(A,B)$ be an instance of $\OV$ with $A=\{a_1,\ldots,a_n\}\subset \{0,1\}^d$ and $B=\{b_1,\ldots,b_n\}\subset \{0,1\}^d$. Furthermore let $v=(v[1],v[2],\ldots,v[d])$ for every $v\in A\cup B$.
We represent $a_i$ and $b_j$ by the strings
\begin{align*}
\psi_A(a_i) &=\psi_A(a_i[1])\psi_A(a_i[2])\cdots \psi_A(a_i[d]),\\
\psi_B(b_j)&=\psi_B(b_j[1])\psi_B(b_j[2])\cdots \psi_B(b_j[d])
\end{align*} 
where $\psi_A(x)=\begin{cases}01\#&\text{ if }x=0,\\ 00\#&\text{ if }x=1\end{cases}$ 
and  
$\psi_B(y)=y\#$ for $y\in\{0,1\}$. 

\begin{claim}\label[Claim]{lem:orth}
$a_i$ and $b_j$ are orthogonal if and only if $\psi_B(b_j)$ occurs in $\psi_A(a_i)$ as a subsequence.
\end{claim}
\begin{proof}
Since $\len{\psi_A(a_i)}_\#= \len{\psi_B(b_j)}_\# = d$ holds, $\psi_B(b_j)$ is a subsequence of $\psi_A(a_i)$ if and only if $\psi_B(b_j[k])$ occurs in $\psi_A(a_i[k])$ for all $k\in [1:d]$. We note that, for $x,y\in \{0,1\}$, $\psi_B(y)$ is absent from $\psi_A(x)$ if and only if $x=y=1$. Hence, $\psi_B(b_j)$ is a subsequence of $\psi_A(a_i)$ if and only if $a_i[k]\cdot b_j[k]=0$ for all $k\in [1:d]$. That is, $\psi_B(b_j)$ is a subsequence of $\psi_A(a_i)$ if and only if $a_i$ is orthogonal to $b_j$. 

\hspace{6cm}(End of the proof of \cref{lem:orth})\end{proof}
With $\psi_A$ and $\psi_B$ we construct two words $W,U\in\{0,1,\#,[,]\}^\ast$ representing $A$ and $B$ as follows: 
\begin{align*}
W &= [\psi_A(\mathbbb{1})][\psi_A(\mathbbb{0})][\psi_A(a_1)][\psi_A(\mathbbb{0})][\psi_A(a_2)][\psi_A(\mathbbb{0})]\ldots [\psi_A(a_n)][\psi_A(\mathbbb{0})][\psi_A(\mathbbb{1})]\\
U &= [\psi_B(\mathbbb{1})][\psi_B(b_1)][\psi_B(b_2)]\ldots [\psi_B(b_n)][\psi_B(\mathbbb{1})]
\end{align*}
where $\mathbbb{0}$ (respectively, $\mathbbb{1}$) stands for the all-zero (respectively, all-one) vector of size $d$. We will occasionally omit $\psi_A$ and $\psi_B$ and call $[\psi_A(v)]$ and $[\psi_B(v)]$ $[v]$-blocks or, more generally, $[\cdot]$-blocks if it is clear whether it occurs in $W$ or in $U$. As such, $[\psi_A(\mathbbb{0})]$ is called a zero-block in the following, while a $[v]$-block is called a non-zero-block if and only if $v\neq \mathbbb{0}$. Thus, the encodings of $U$ and $W$ include two $[\mathbbb 1]$-blocks at the beginning and at the end of $U$ and $W$, respectively, and a zero-block after each non-zero-block of $W$ (excluding the $[\mathbbb 1]$-block at the end).

\begin{remark}\label[Remark]{rem:brack}
Since $U$ starts with the letter $[$, if $U$ is a subsequence of any bounded range of length $|W|$ of $W^2$, then $U$ is a subsequence of a bounded range of length $|W|$ of $W^2$ starting with $[$.
\end{remark}

Next we show that the instance $(A,B)$ of $\OV$ is accepted if and only if the instance $u=U$, $w=W^2$ and $p=|W|$ of $\pSubSeqMatch$ is accepted.

\begin{claim}\label[Claim]{lem:orth-subs}
There are orthogonal vectors $a_i\in A$ and $b_j\in B$ if and only if $U$ is a $|W|$-subsequence of $W^2$.
\end{claim}
\begin{proof}
Firstly, suppose $a_i$ and $b_j$ are orthogonal. Then $[\psi_B(b_j)]$ occurs in $[\psi_A(a_i)]$. Since $|W| = (2n+3)\cdot|[\psi_A(v)]|$ for any $v\in\{0,1\}^d$, we can choose a bounded range (until the end of this proof the reader may safely assume every bounded range to be a bounded range of length $|W|$ of $W^2$) containing $2n+2$ $[\cdot]$-blocks around $[\psi_A(a_i)]$. Furthermore, every bounded range starting with $[$ contains exactly $n+1$ zero-blocks. Hence, we choose a bounded range containing $j$ zero-blocks to the left of $[\psi_A(a_i)]$ and $n-j+1$ zero-blocks to the right of $[\psi_A(a_i)]$. If $j\leq i$, we match the $[b_j]$-block in $U$ against the first occurrence of the $[a_i]$-block in $W^2$ and choose the bounded range starting at the first occurrence of $[\psi_A(a_{i-j})]$. If $j>i$, we match the $[b_j]$-block in $U$ against the second occurrence of the $[a_i]$-block in $W^2$ and choose the bounded range starting at the first occurrence of $[\psi_A(a_{n+i-j+1})]$. In both cases there is one zero-block to the left (respectively, right) of $[\psi_A(a_i)]$ for one $[\mathbbb 1]$-block and each $[b_k]$-block for $1\leq k<j$ (respectively, $j<k\leq n$). Hence, $U$ is a $|W|$-subsequence of $W^2$.

For the inverse implication, suppose that $a_i$ is not orthogonal to $b_j$ for all $1\leq i,j\leq n$. 
By \cref{lem:orth}, no $[b_j]$-block occurs in any $[a_i]$-block as a subsequence, hence the $[\cdot]$-blocks of $U$ only occur in zero-blocks of $W^2$. 
By \cref{rem:brack}, it suffices to show that $U$ does not occur in a bounded range starting with $[$. Each of those has exactly $n+1$ zero-blocks but $U$ has $n+2$ non-zero-blocks. Thus, $U$ is $|W|$-absent from $W^2$. \hspace{2.5cm} (End of the proof of \cref{lem:orth-subs})\end{proof}

Finally, we note that $|W|=(2n+3)(3d+2)\in \mathcal O(n\cdot poly(d))$ and $|U|=(n+2)(2d+2) \in \mathcal O(n\cdot poly(d))$ and so an algorithm deciding $\pSubSeqMatch$ in time $\mathcal O(|W|^h|U|^g)$, with $h+g=2-\epsilon$, could also be used to decide $\OV$ in time $\mathcal O(n^{2-\epsilon}poly(d))$, which is not possible by $\OVH$. Hence, \cref{thm:lowerBoundSubseq} holds.
\end{proof}

\section{Analysis Problems}
This section covers several decision problems regarding the set $\pSubseq_k$ defined for an input word. More precisely, we approach problems related to the universality of this set, the equivalence of the respective sets for two words, as well as problems related to minimal and shortest absent subsequences w.\,r.\,t. these sets. For these problems we give respective hardness results and fine-grained  conditional lower bounds. We start with the following problem.

\begin{problemdescription}\label{kpNonUniv}
  \problemtitle{$k$-Non-Universality in Bounded Range, $\kpNonUniv$}
  \probleminput{A word $w$ over $\Sigma$ and integers $k, p$, with $\len{w} = n$, and $k\leq p\leq n$.}
  \problemquestion{Decide whether $\pSubseq_{k}(w) \neq \Sigma^k$.}
\end{problemdescription}

Let us observe that if $|w|<k\sigma$, where $\sigma = |\Sigma|$, then we can trivially conclude by the definition of the universality of a word that $\Subseq_k(w)\neq \Sigma^k$, so $\pSubseq_{k}(w) \neq \Sigma^k$ as well. Therefore, to avoid the trivial inputs of $\kpNonUniv$ we will assume that $|w|\geq k\sigma$. 



To show that $\kpNonUniv$ is NP-complete,
we first examine a related problem given in \cite{Manea2013}.
We need some preliminaries.
A partial word over an alphabet $\Sigma$ is a string from $(\Sigma\cup\{\Diamond\})^*$.
In such a partial word $u$,
we have defined positions (those positions $i$ for which $u[i]\in \Sigma$)
and undefined positions (those positions $i$ for which $u[i]=\Diamond$);
intuitively, while the letters on the defined positions are fixed,
the $\Diamond$ can be replaced by any letter of the alphabet $\Sigma$
and, as such, a partial word actually describes a set of (full) words over $\Sigma^*$,
all of the same length as $u$.
This is formalized as follows.
If $u$ and $v$ are partial words of equal length,
then $u$ is contained in $v$, denoted by $u \subseteq v$, if $u[i] = v[i]$,
for all defined positions $i$ (i.\,e., all positions $i$ such that $u[i]\in \Sigma$).
Moreover, the partial words $u$ and $v$ are compatible,
denoted by $u\uparrow v$, if there exists a full word $w$ such that $u \subseteq w$ and $v \subseteq w$. 
In this framework, we can define the problem $\pWords$,
which we will use in our reductions.

\begin{problemdescription}
  \problemtitle{Partial words non-universality, $\pWords$}
  \probleminput{A list of partial words $S = \{w_1, \ldots, w_k\}$ over $\{0,1\}$, each partial word having the same length $L$}
  \problemquestion{Decide whether there exists a word $v \in \{0,1\}^L$ such that $v$ is not compatible with any of the partial words in $S$.}
\end{problemdescription}

The first part of the following result was shown in \cite{Manea2013} via a reduction from $3$-$\SatProb$, and it can be complemented by a conditional lower bound.

\begin{theorem}\label[Theorem]{thm:pWordsCNFSat}
	$\pWords$ is NP-complete and cannot be solved in subexponential time $2^{\smallO(L)} \poly(L,n)$ unless $\ETH$ fails.
\end{theorem}
\begin{proof}
	A full proof of the first statement is given in \cite{Manea2013}
	but we will give here some overview and essential key parts from it to see how the second part of the statement follows.
	The main idea is based on a reduction from the NP-complete problem $3$-$\SatProb$ (where the conjunctive normal form is additionally, but also trivially, restricted in the way that every clause cannot contain both a variable and its negation at the same time).
	
	Considering an instance of the $3$-$\SatProb$ problem,
	one can associate a corresponding instance of $\pWords$
	where the number of partial words for $\pWords$ equals the number of clauses from $3$-$\SatProb$,
	and the number of variables equals the length of the partial words. 
	
	Assume that the instance of the $3$-$\SatProb$ problem consists in the Boolean formula $F$ which is the conjunction of the clauses $\{c_1, \ldots, c_m\}$ over a set of variables $V = \{v_1, v_2, \ldots, v_n\}$. Clearly, $F$ is satisfiable if and only if there exists an assignment of the variables from $V$ which makes all clauses $\{c_1, \ldots, c_m\}$ equal to $1$, and this corresponds to an assignment of the variables from $V$ which makes all clauses $\{\neg c_1, \ldots, \neg c_m\}$ equal to $0$. 
	
	We now define an instance of the $\pWords$ problem with $n$ words $w_1,\ldots, w_n$. In general, we assign the value $1$ to position $j$ in partial word $w_i$ if variable $v_j$ appears in clause $\neg c_i$ in its unnegated form,
	and we assign the value $0$ to position $j$ in partial word $w_i$ if $\neg v_j$ appears in clause $\neg c_i$, i.\,e., the negation of variable $v_j$ appears in $\neg c_i$.  Otherwise, we assign the value $\Diamond$ to position $j$ in partial word $w_i$. 
	
	By this construction, we obtain that the instance  of the $3$-$\SatProb$ problem is satisfiable if and only if there exists a word $v$ of length $L$ which is not compatible to any of the words $w_i$. Therefore, $\pWords$ is NP-hard (and the completeness now follows trivially). Moreover, as $L$, the length of the partial words in the instance of our problem, equals the number of variables in the $3$-$\SatProb$, the conditional lower bound holds as well. 
\end{proof}
	
Based on the hardness of $\pWords$, we continue by showing that the k-Non-Universality in Bounded Range Problem is also NP-hard.

\begin{theorem}\label[Theorem]{thm:univ-np-hard}
	$\kpNonUniv$ is NP-hard and cannot be solved in subexponential time $2^{\smallO(k)} \poly(k,n)$ unless $\ETH$ fails.
\end{theorem}

\begin{proof}
	Given an instance of $\pWords$, that is, given a set $S = \{w_1, \ldots, w_k\}$ of partial words,
	each of length $L$,	
	we construct an instance of $\kpNonUniv$ as follows. \looseness=-1
	For every $w_i \in S$,
	let $u_i = u_i^1 u_i^2 \cdots u_i^L$ with
	\[
		u_i^j = \begin{cases} 0\# & \text{if } w_i[j] = 0,\\
					  1\# & \text{if } w_i[j] = 1,\\
					  01\# & \text{if } w_i[j] = \Diamond ,
			\end{cases}
	\]
	\noindent
	for all $j \in [L]$. Now, we define the words
	\begin{align*}
		V &= \#^{2L} (001101\#^{2L})^{L-1}, \\
		U &= \#^{4L^2} u_1 \#^{4L^2} u_2 \#^{4L^2} \cdots \#^{4L^2} u_k \#^{4L^2}.
	\end{align*}
	Finally, we let $W=VU$. Clearly, $|V|=2L+(6+2L)(L-1)= 2L^2+6L - 6$.
	Note that for the definition of $U$ the number of separating $\#$-symbols only needs to be large enough with respect to a window size of $\len V$ to effectively isolate the factors $u_1, u_2, \ldots, u_k$, so it could be possible that the exponent $4L^2$ can be lowered to a certain extent.
	For the instance of $\kpNonUniv$, we set $k= 2L$ and $p= \len V$.
	
	Note that every word of length $2L$ over the alphabet $\{0,1,\#\}$ is a subsequence of $V$
	except for words $a_1 \# a_2 \# \cdots \# a_L \#$ with $a_i \in \{0,1\}$.
	Thus, due to the relatively large prefix $\#^{4L^2}$ of $U$, the only possibility for a word of the form $a_1 \# a_2 \# \cdots \# a_L \#$ with $a_i \in \{0,1\}$ to be a subsequence of $W$ is
	when it occurs in $U$ as a subsequence. Moreover, for a word $a_1 \# a_2 \# \cdots \# a_L \#$ to occur as a subsequence in a range of length $|V|$ of $U$, it must occur as a subsequence of one of the words $u_i$, with $i\in [1:k]$ (again, due to the large factors $\#^{4L^2}$ of $U$ separating the words $u_i$). This means that $W$ is $k$-universal if and only if all words from $\Sigma^k$ of the form $a_1 \# a_2 \# \cdots \# a_L \#$ with $a_i \in \{0,1\}$ occur in $U$.
	
	Now, let us observe that a word $x'=a_1 a_2 \cdots a_L  \in \{0,1\}^L$ is compatible with one of the words $w_i$ if and only if $x=a_1 \# a_2 \# \cdots a_L \# $ is a subsequence of $u_i$. For the implication from left to right note that if $a_j=w_i[j]\in \{0,1\}$, then $a_j\# =u_i^j\#$. Moreover, if $w_i[j]=\Diamond$, then $a_j\#$ is a subsequence of $u_i^j=01\#$. For the converse, assume $x=a_1 \# a_2 \# \cdots a_L \# $ is a subsequence of $u_i$. It follows that $a_j\#$ must be a subsequence of $u_i^j$. If $u_i^j=\ell \#$ for $\ell\in \{0,1\}$, then we have that $a_j\# = \ell \#$ and $\ell= w_i[j]$. If $u_i^j=01 \#$, then we have that $w_i[j]=\Diamond$, so $a_j$ is compatible with $w_i[j]$. Clearly, $x'$ is compatible with $w_i$. 
	
	Therefore, $\pSubseq_{k}(w) \neq \Sigma^k$ holds
	if and only if there exists a word $x=a_1 \# a_2 \# \cdots \# a_L \#$, for some $a_i \in \{0,1\}$, that is not a subsequence of any of the $u_1$, $\ldots$, $u_k$.
	And if there exists such a word $x$,
	then and only then there exists a word $x'=a_1 a_2 \cdots a_L  \in \{0,1\}^L$ such that $x'$ is not compatible with any of the words from~$S$. So our reduction is valid. The instance of $\pWords$ is accepted if and only if the instance of $\kpNonUniv$ is accepted.
	
	To analyse the complexity of this reduction, note that 
	we have constructed several words in time polynomial in $k$ and $L$. 
	
	This means that we now have a valid reduction from $\pWords$ to $\kpNonUniv$, which works in $poly(L,n)$ time.
	
	Ultimately, we have a reduction from $3$-$\SatProb$ to $\pWords$ (from \cref{thm:pWordsCNFSat}) 
	and a reduction from $\pWords$ to $\kpNonUniv$.
	As both reductions can be performed in polynomial time and in the reduction from $\pWords$ to $\kpNonUniv$ we have $k=2L$,
	we also obtain an $\ETH$ lower bound for $\kpNonUniv$. That is, it is immediate that $\kpNonUniv$ cannot be solved in subexponential time $2^{\smallO(k)} \poly(k,n)$ unless $\ETH$ fails.
\end{proof}

This hardness result is complemented by the following algorithmic result.
\begin{remark}
Note that $\kpNonUniv$ can be trivially solved in $\mathcal{O}(\sigma^k \poly(k,n))$ by a brute-force algorithm that simply checks for all words from $\Sigma^k$ whether they are in $\pSubseq_{k}(w)$. For $\sigma\in \mathcal{O}(1)$, this algorithm runs in $\mathcal{O}(2^k \poly(k,n))$.
\end{remark}

Now, when looking at a related analysis problem, we consider two different words $w$ and $v$,
and we want to check whether both words are equivalent w.\,r.\,t. their respective sets of $p$-subsequences of length $k$.

\begin{problemdescription}\label{p:kpNonEquiv}
  \problemtitle{$k$-Non-Equivalence w.\,r.\,t. Bounded Ranges, $\kpNonEquiv$}
  \probleminput{Two words $w$ and $v$ over $\Sigma$ and integers $k, p$, with $\len{w} = n$, $\len v = m$, and $k \leq p \leq m, n$.}
  \problemquestion{Decide whether  $\pSubseq_{k}(w) \neq \pSubseq_{k}(v)$.}
\end{problemdescription}


We can now state the following theorem.

\begin{theorem}
	$\kpNonEquiv$ is NP-hard and cannot be solved in time $2^{\smallO(k)} \poly(k,n,m)$ unless $\ETH$ fails.
\end{theorem}
\begin{proof}
We use the same reduction from $\pWords$ as in the proof of \cref{thm:univ-np-hard}. The difference is that after constructing the word $W$ as in the respective reduction, we also construct the word $W'= V \#^{4L^2} (01\#)^L \#^{4L^2}$. By the same arguments as in the proof of the respective theorem, we have that $\pSubseq_k(W')=\{0,1,\#\}^k$ (where $k=2L$ and $p=|V|$). Therefore, the instance of $\pWords$ is accepted if and only if $\pSubseq_k(W)=\{0,1,\#\}^k$ if and only if $\pSubseq_k(W)=\pSubseq_k(W')$. The conclusion follows now immediately.
\end{proof}
%
%
%

Again, a matching upper bound is immediate.
\begin{remark}
$\kpNonEquiv$ can be trivially solved in $\mathcal{O}(\sigma^k \poly(k,n))$ by a brute-force algorithm that looks for a word from $\Sigma^k$ which separates $\pSubseq_{k}(w) $ and $\pSubseq_{k}(v) $. For $\sigma\in \mathcal{O}(1)$, this algorithm runs in $\mathcal{O}(2^k \poly(k,n))$.
\end{remark}

A natural problem arising in the study of the sets $\pSubseq_{k}(w)$, for $k\leq p\leq |w|$, is understanding better which are the strings missing from this set. To that end, we have introduced in \cref{sec:defs} the notions of shortest and minimal absent $p$-subsequences, $\psas$ and $\pmas$, respectively.

We first focus on shortest absent subsequences in bounded ranges.
\begin{problemdescription}\label{p:nonPSAS}
  \problemtitle{Non-Shortest Absent Subsequence w.\,r.\,t. Bounded Ranges, $\pNonPSAS$}
  \probleminput{Two words $w$ and $v$ over $\Sigma$ and integer $p$, with $\len{w} = n$, $\len{v}=m$, and $m\leq p\leq  n$.}
  \problemquestion{Decide whether $v$ is not a $\psas$ of $w$..}
\end{problemdescription}

We can show the following result.

\begin{theorem}
	$\pNonPSAS$ is NP-hard and cannot be solved in subexponential time $2^{\smallO(k)} \poly(k,n,m)$ unless $\ETH$ fails.
\end{theorem}
\begin{proof}
We begin our proof with an observation. Assuming that word $v$ (the subsequence in question) with $\len{v} = m$ is already confirmed to be an absent $p$-subsequence, then $v$ is not a $\psas$ if and only if $\xSubseq{p}_{m-1} (w) \neq \Sigma^{m-1}$ holds.
	
	For the hardness reduction, we use again the $\pWords$ problem
	and give a reduction similar to the one from the proof of \cref{thm:univ-np-hard}.
	
	Let $S = \{w_1, \ldots, w_k\}$ be an instance of the partial words problem,
	with each of the words $w_1, \ldots, w_k$ from set $S$ having length $L$.
	Recall the construction of an instance of $\kpNonUniv$.
	
	For every $w_i \in S$,
	construct $u_i = u_i^1 u_i^2 \cdots u_i^L$ where, for all $j \in [L]$,
	\[
	u_i^j = \begin{cases} 0\# & \text{if } w_i[j] = 0\\
		1\# & \text{if } w_i[j] = 1\\
		01\# & \text{if } w_i[j] = \Diamond
	\end{cases}
	\]
	
	\begin{align*}
		V &= \#^{2L} (001101\#^{2L})^{L-1}, \\
		U &= \#^{4L^2} u_1 \#^{4L^2} u_2 \#^{4L^2} \cdots \#^{4L^2} u_k \#^{4L^2}, \\
		W &= VU.
	\end{align*}
	
	For the problem instance of $\psas$,
	let $w = W$, $p=\len V$ and $v = (0\#)^{L}0$,
	and note that $v$ has length $m = 2L + 1$.
	Now, it is immediate that (because of the large $\#$-blocks occurring in $W$) $v$ is an absent $p$-subsequence of $W$.
	Therefore, it holds that $v$ is a $\psas$ of $w$ if and only if $w$ is $(m-1)$-universal,
	or in other words if $W$ is $2L$-universal. We have seen that this is equivalent to checking whether the input instance of $\pWords$ is accepted or not. 
	
	At this point, we have a valid reduction from $\pWords$ to $\pNonPSAS$, which can be implemented in time $\poly(L,n,m)$.
	Therefore, our claim holds, and moreover we again obtain an $\ETH$ conditional lower bound.
	Consequently, $\pNonPSAS$ cannot be solved in subexponential time $2^{\smallO(k)} \poly(l,n,m)$ unless $\ETH$ fails.
	Thus, our claim holds.
	\end{proof}

Now let us examine minimal absent subsequences in bounded ranges, $\pmas$. First we give an algorithm to check whether a string is a $\pmas$ of another string.

\begin{problemdescription}
  \problemtitle{Minimal Absent Subsequences w.\,r.\,t. Bounded Ranges, $\pMAS$}
  \probleminput{Two words $w$ and $v$ over $\Sigma$ and integer $p$, with $\len{w} = n$, $\len{v}=m$, and $m\leq p\leq  n$.}
  \problemquestion{Decide whether $v$ is a $\pmas$ of $w$.}
\end{problemdescription}


In this case, we obtain a polynomial time algorithm solving this problem.
\begin{theorem}\label[Theorem]{thm:algpmas}
$\pMAS$ can be solved in time $\mathcal{O}(nm)$, where $\len v = m, \len w = n$.
\end{theorem}
\begin{proof}
As in the case of the algorithm solving $\pSubSeqMatch$ in $\mathcal{O}(mn)$ time, presented in the proof of \cref{thm:algoMatching}, we will propose here an algorithm which works in a streaming fashion w.\,r.\,t. the word $w$. In this case, our algorithm scans the letters of $w$ one by one, left to right, while maintaining the following information. On the one hand, our algorithm checks if $v$ occurred as a subsequence in any of the factors $w[t-p+1:t] $ scanned so far (so, the strings which occur as suffixes of length $p$ of the stream during the scan). On the other hand, it maintains data structures allowing us to keep track of the factors $v[1:i-1]v[i+1:m]$ of $v$ (with $i\in [1:m]$) which occurred as subsequences in the factors $w[t-p+1:t] $ scanned so far. After all the letters of $w$ are scanned, the algorithm decides that $v$ is a $\pmas$ of $w$ if and only if $v$ did not occur in any of the windows $w[t-p+1:t] $ for $t\in [p:n]$ and each of the factors $v[1:i-1]v[i+1:m]$ of $v$ (with $i\in [1:m]$) occurred as subsequence in at least one of the factors $w[t-p+1:t] $, for $t\in [p:n]$. 

In our algorithm, we maintain $\sigma$ lists of letters: for each $a\in \Sigma$, we have a list $L_a$. The following invariant will be maintained during the algorithm, for all $a\in \Sigma$: after letter $w[t]$ is processed by our algorithm, $L_a$ contains the positions $i_1,i_2,\ldots,i_\ell$ of $w$ where $a$ occurs inside the bounded range $w[t-p+1:t]$, ordered increasingly. 

Moreover, we maintain $2m$ pointers to elements of these lists. On the one hand, we have the pointers $p_i$, with $i\in [1:m]$, where $p_i$ points to the position $j_{i}$ (of $w$) stored in the list $L_{v[i]}$. On the other hand, we have the pointers $\ell_i$, with $i\in [1:m]$, where $\ell_i$ points to the position $g_{i}$ (of $w$) stored in the list $L_{v[i]}$. The invariant property (holding after $w[t]$ was processed) of these pointers, and the positions they point, is that, for all $i\in [1:m]$, $w[t-p+1:j_i]$ (respectively, $w[g_i:t]$) is the shortest prefix (respectively, suffix) of $w[t-p+1:t]$ which contains $v[1:i]$ (respectively, $v[i:m]$) as subsequence; $p_i$ (respectively, $\ell_i$) is undefined if such a prefix (respectively, suffix) does not exist. 

Intuitively, the positions $j_i$, for $i\in [1:m]$, can be seen as the positions where the letters $v[i]$ of $u$ are identified in $w[t-p+1:t]$, respectively, when the classical subsequence matching greedy algorithm (see, e.\,g., \cite{Kosche2021}) is used to identify the letters of $u$ greedily left to right in $w[t-p+1:t]$. That is, $j_i$ is the position where we find $v[i]$ first (in the order from left to right), after position $j_{i-1}$. Similarly, the positions $g_i$, for $i\in [1:m]$, can be seen as the positions where the letters $v[i]$ of $u$ are identified in $w[t-p+1:t]$, respectively, when the classical subsequence matching greedy algorithm (see, e.\,g., \cite{Kosche2021}) is used to identify the letters of $u$ greedily right to left in $w[t-p+1:t]$. That is, $g_i$ is the position where we find $v[i]$ first when going right to left from position $g_{i+1}$. 

Alongside the lists and set of pointers, we also maintain a set $S\subseteq [1:m]$, with the invariant property that, after $w[t]$ was processed, we have that $i\in S$ if and only if $v[1:i-1]v[i+1:m]$ is a subsequence of some factor $w[j-p+1:j]$ with $j\leq t$. 

Let us now describe the core of our algorithm: how are the lists $L_a$, for $a\in \Sigma$, and the pointers $p_i$ and $\ell_i$, for $i\in [1:m]$ maintained.

We initially define all the lists to be empty and leave all pointers undefined. We define the set $S$, which is initially empty. 

When reading $w[t]$, for $t$ from $1$ to $n$ we do the following steps. 
\begin{enumerate}
\item We insert position $t$ at the end of the list $L_{w[t]}$. 
\item If $j_1=t-p$ (i.\,e., the pointer $p_1$ points to position $t-p$), then move $p_1$ one position to the right in list $L_{v[1]}$, and update $j_1$ accordingly; if $t-p$ was the single element of the list $L_{v[1]}$, then set $p_1$ to be undefined and $j_1\gets +\infty$. Further, if $p_1$ was updated, for all $i>1$ in increasing order, move $p_i$ to the right in list $L_{v[i]}$ and stop as soon as they point to a position $j_i>j_{i-1}$; if such a position does not exist in $L_{v[i]}$, then $p_i$ is set to be undefined and $j_i\gets +\infty$. If, at some point, a pointer $p_i$ is set to be undefined, then all the pointers $p_d$, with $d>i$, are set to be undefined too, and $j_d\gets +\infty$.
 
\item If $w[t]=v[m]$, then set $\ell_m$ such that it points to position $t$ in $L_{v[m]}$ and set $g_m\gets m$. Further, if $\ell_m$ was updated, then for all $i$ from $m-1$ to $1$ in decreasing order, move $\ell_i$ to the right in list $L_{v[i]}$ (starting from its current position, or from the first position in the list, if $\ell_i$ was undefined) and stop as soon as we reach the last element of that list or an element whose right neighbour in the respective list is a position $g'>g_{i+1}$; if such an element does not exist in $L_{v[i]}$, then $\ell_i$ is set to be undefined and $j_1\gets +\infty$. If $\ell_i$ points at position $g_i=t-p$, then $\ell_i$ is set to be undefined and $j_1\gets +\infty$. If, at some point, a pointer $\ell_i$ is set to be undefined, then all the pointers $\ell_d$, with $d<i$, are set to be undefined too, and $j_d\gets +\infty$.

\item We remove position $t-p$ from the beginning of the list  $L_{w[t-p]}$.

\item Once these update operations are completed, we check whether there exists some $i\in[1:m-1]$ such that $j_i<g_{i+1}$. If yes, we conclude that $v$ occurs as a subsequence in $w[t-p+1:t]$, so it cannot be a $\pmas$ of $w$ (and we stop the algorithm). 

\item Further, we check for all $i\in[1:m]$ whether $j_{i-1}<g_{i+1}$ (where $j_0=0$ and $g_{m+1}=m+1$); if $j_{i-1}<g_{i+1}$ holds, then $v[1:i-1]v[i+1:m]$ occurs as a subsequence of $w[t-p+1:t]$, so we insert $i$ in $S$. 
\end{enumerate}

If the process above was successfully executed for all $t\in [1:n]$, we check if $S=[1:m]$, and if yes, we decide that $v$ is a $\pmas$ of $w$. Otherwise, we decide that $v$ is not a $\pmas$ of $w$.

Let us now show the correctness of our approach.  

Firstly, we show that the invariants are preserved at the end of the processing of letter $w[t]$ (provided that they were true before $w[t]$ was processed). 

Clearly, these invariants are true initially. 

Now, assume that the respective invariant properties were true before $w[t]$ was processed. Then it is immediate that, after running the processing for $w[t]$, for all $a\in \Sigma$, we have that $L_a$ contains the positions $i_1,i_2,\ldots,i_\ell$ of $w$ where $a$ occurs inside the bounded range $w[t-p+1:t]$, ordered increasingly. 

The fact that the invariant properties regarding the pointers $p_i$ and $\ell_i$ and the positions pointed by them are still true is not that immediate. Let us first discuss the pointers $p_i$, for $i\in [1:m]$, and the positions $j_i$ they respectively point at. Firstly, after position $t$ is inserted in $L_{w[t]}$, we have that for all $i\in [1:m]$, $w[t-p:j_i]$ is the shortest prefix of $w[t-p:t]$ which contains $v[1:i]$ as subsequence, and $p_i$ is undefined if such a prefix does not exist. Unless $w[t-p]=v[1]$, we clearly also have that for all $i\in [1:m]$, $w[t-p+1:j_i]$ is the shortest prefix of $w[t-p+1:t]$ which contains $v[1:i]$ as subsequence, and $p_i$ is undefined if such a prefix does not exist. So, let us assume that $w[t-p]=v[1]$. In that case, we need to update the pointers. It is easy to see that, in order to preserve the invariants, a correct way to compute the positions $j_i$, for $i\in [1:m]$, is to look first for the leftmost position $j_1$ where $v[1]$ occurs inside $w[t-p+1:t]$, and then look, for $i$ from $2$ to $m$, for the leftmost position $j_i$ where $v[i]$ occurs inside $w[t-p+1:t]$ after $j_{i-1}$. And, this is exactly what we do in the second step of the update process described above. Thus, the invariant property of the pointers $p_i$ is preserved.

A similar approach holds for the case of the pointers $\ell_i$ and the positions $g_i$ they respectively point at. In this case, after position $t$ is inserted in $L_{w[t]}$, if $w[t]\neq v[m]$, we have that, for all $i\in [1:m]$, $w[g_i:t]$ is the shortest suffix of $w[t-p:t]$ which contains $v[i:m]$ as subsequence, and $j_i$ is undefined if such a suffix does not exist. So, after removing position $t-p$ from $L_{w[t-p]}$ and setting $\ell_i$ to undefined and $g_i=+\infty$ if $g_i=t-p$ for some $i$, we clearly also have that, for all $i\in [1:m]$, $w[g_i:t]$ is the shortest suffix of $w[t-p:t]$ which contains $v[i:m]$ as subsequence, and $j_i$ is undefined if such a suffix does not exist. Assume now that $w[t]= v[m]$.  Clearly, in order to preserve the invariants, a correct way to compute the positions $g_i$, for $i$ from $m$ to $1$, is to set $g_m\gets t$, and then look, for $i$ from $m-1$ to $1$, for the rightmost position $g_i$ where $v[i]$ occurs inside $w[t-p+1:t]$ before $g_{i-1}$. Once again, this is exactly what we do in the second step of the update process described above. Thus, the invariant property of the pointers $\ell_i$, and the corresponding positions $g_i$, is preserved.

We additionally note that $v$ occurs as a subsequence of $w[t-p+1:t$ if and only if there exists $i$ such that the shortest prefix of $w[t-p+1:t]$ which contains $v[1:i]$ as a subsequence does not overlap with the shortest suffix of $w[t-p+1:t]$ which contains $v[i+1:m]$ as a subsequence, i.\,e., $j_{i}<g_{i+1}$. Therefore the decision reached in Step $5$ of the above update process is correct.

Finally, it is clear that $v[1:i-1]v[i+1:m]$ occurs in $w[t-p+1:t]$ if and only if the shortest prefix of $w[t-p+1:t]$ which contains $v[1:i-1]$ as a subsequence does not overlap with the shortest suffix of $w[t-p+1:t]$ which contains $v[i+1:m]$ as a subsequence, i.\,e., $j_{i-1}<g_{i+1}$. Therefore, the invariant property of $S$ is also preserved by the computation in Step $6$.

In conclusion, the invariant properties of the maintained data structures are preserved by our update process. So, either the algorithm correctly reports that $v$ is not a $\pmas$ of $w$, while updating the structures for some $t$, because it detected that $v\in \pSubseq(w)$ or $v\notin \pSubseq(w)$ and the algorithm reports (after doing all the updates, for all positions $t$ of $w$) that $v$ is a $\pmas$ of $w$ if and only if $S=[1:m]$ (i.\,e., all subsequences of length $m-1$ of $v$ are in $\pSubseq(w)$). 

Therefore our algorithm works correctly. Now, let us analyse its complexity. 

We do $n$ update steps and, in each of them, we update the lists $L_a$, for $a\in \Sigma$, the set $S$ (by simply adding some new elements), and update the pointers $p_i$ and $\ell_i$, and the corresponding positions $j_i$ and $g_i$, respectively, for $i\in [1:m]$. The set $S$ can be implemented using a characteristic array consisting in $m$ bits, so that all insertions in $S$ can be done in $\mathcal{O}(1)$ time, and checking whether $S=[1:m]$ can be done in $\mathcal{O}(m)$ time. Now, we focus on the time needed for the updates of the lists and pointers. As each position of the word is scanned at most once by each pointer (they move left to right only), and each position of the word is inserted once and removed at most once in exactly one of the lists, the overall time used to process these lists is $\mathcal{O}(mn + n)=\mathcal{O}(mn)$. In conclusion, the time complexity of the entire algorithm is $\mathcal{O}(mn)$. This concludes our proof.
\end{proof}

%

Similarly to the proof of \cref{thm:algoMatching}, we propose an algorithm which can be seen as working in the sliding window model, with window of fixed size $p$. If, as in the case of the discussion following \cref{thm:algoMatching}, we assume $u$ (and $m$) to be constant, we obtain a linear time algorithm. However, its space complexity, measured in memory words, is $\mathcal{O}(p)$ (as we need to keep track, in this case, of entire content of the window). In fact, when $m$ is constant, it is easy to obtain a linear time algorithm using $\mathcal{O}(1)$ memory words (more precisely, $\mathcal{O}(\log p)$ bits of space) for this problem: simply try to match $u$ and all its subsequences of length $(m-1)$ in $w$ simultaneously, using the algorithm from \cref{thm:algoMatching}. Clearly, $u$ is a $\pmas$ if and only if $u$ is not a subsequence of $w$, but its subsequence of length $m-1$ are. However, the constant hidden by the ${\mathcal O}$-notation in the complexity of this algorithm is proportional with $m^2$. It remains open whether there exists a (sliding window) algorithm for  $\pMAS$ both running in ${\mathcal O}(mn)$ time (which we will show to be optimal, conditional to $\OVH$) and using only ${\mathcal O}(\log p)$ bits (which is also optimal for sliding window algorithms).

\begin{theorem}\label[Theorem]{thm:logSpaceLowerBoundpMas}
	The $\pMAS$ problem cannot be solved by a sliding window algorithm using ${o}(\log p)$ bits. 
\end{theorem}
\begin{proof}
	This can be shown by a reduction from $\pSubSeqMatch$. The results of \cite{GanardiHL16} show that $\pSubSeqMatch$ cannot be solved using ${o}(\log p)$ bits. 
	
	So, assume for the sake of contradiction that $\pMAS$ has a sliding window solution which uses  ${o}(\log p)$ bits. 
	
	Let $p\in\mathbb N$ and $u,w\in\Sigma^*$ be an instance of $\pSubSeqMatch$. For $1\leq i\leq m=|u|$, we define $u_i=u[1:i-1]u[i+1:m]$ to be the subsequence we obtain from $u$ by deleting its \nth i letter. We let \(w'=\$^{p+1}u_1\$^{p+1}u_2\$^{p+1}\ldots\$^{p+1} u_m\), where $\$$ is an additional letter not occurring in $\Sigma$, and $v=ww'$.
	
	Then, every $u_i$ is a $p$-subsequence of $v$, hence every proper subsequence of $u$ is a $p$-subsequence of $v$. So, $u$ is a $\pmas$ of $v$ if and only if it is $p$-absent from $v$. Since $u$ is $p$-absent from $w'$ (due to the $\$^{p+1}$ separators and the fact that $u$ contains no $\$$ symbol) and no proper suffix of $u$ occurs as a prefix of $w'$ (again, due to the large-enough $\$^{p+1}$ separator), $u$ is $p$-absent from $v$ if and only if $u$ is $p$-absent from $w$. In conclusion, $u$ is a $\pmas$ of $v$ if and only if $u$ is $p$-absent from $w$.
	
	Therefore, checking whether $u$ is $p$-absent from $w$ (or, alternatively, $u$ is a $p$-subsequence of $w$) it is equivalent to checking whether $u$ is a $\pmas$ of $w$. Our assumption that $\pMAS$ has a sliding window solution which uses  ${o}(\log p)$ bits would therefore imply that $\pSubSeqMatch$ can be solved by a sliding window algorithm using ${o}(\log p)$ bits, a contradiction.
\end{proof}

Complementing the discussion above, we now show that it is possible to construct in linear time, for words $u,w$ and integer $p\in\mathbb N$, a string $w'$ such that deciding whether $u$ is a $\pmas$ of $w'$ is equivalent to deciding whether $u$ is a $p$-subsequence of $w$, so solving $\pSubSeqMatch$ for the input words $u$ and $w$. Hence, the lower bound from Theorem \ref{thm:lowerBoundSubseq} carries over,  and the algorithm in \cref{thm:algpmas} is optimal (conditional to $\OVH$) from the time complexity point of view. \looseness=-1

\begin{theorem}
$\pMAS$ cannot be solved in time $\mathcal{O}(n^h m^g)$ where $h+g= 2-\epsilon$ with $\epsilon>0$, unless $\OVH$ fails. 
\end{theorem}
\begin{proof}
We reduce from $\pSubSeqMatch$. From the proof of \cref{thm:lowerBoundSubseq}, we have that
$\pSubSeqMatch$ cannot be decided in time $\mathcal{O}(n^h m^g)$, where $h+g= 2-\epsilon$ with $\epsilon>0$, unless $\OVH$ fails,  even if the input words are over an alphabet $\Sigma$ with $|\Sigma|=5$. Let $p_0\in\mathbb N$ and $u,w\in\Sigma^*$, with $n=|w|$ and $m=|u|$, be an input instance for $\pSubSeqMatch$ with $|\Sigma|=\sigma \in \mathcal{O}(1)$. Assume $\$$ is a symbol not occurring in $\Sigma$.

For $1\leq i\leq m-1$, we define $u_i$ to be a word of length $|\Sigma|$, where each letter of $\Sigma$ occurs exactly once, and $u_i$ ends with $u[i]$ (i.e., $u_i$ is a permutation of $\Sigma$, with $u[i]$ on its last position). Now, we define $w'=
\phi(w)\$^{3p_0\sigma} u_1u_2\ldots u_{m-1}$, where $\phi$ is a morphism which simply replaces each letter $a\in \Sigma$ by $a\$^{\sigma-1}$. Note that $|w'|=\sigma|w|+3p_0\sigma + \sigma (m-1) \in \mathcal{O}(|w|).$

We define the instance of $\pMAS$ where we want to check whether $u$ is a $\pmas$ of $w'$, for $p=p_0\sigma$. We show that this instance of $\pMAS$ is accepted if and only if the instance of $\pSubSeqMatch$ defined by $u,w$, and $p_0$ is rejected (i.e., $u$ is not a $p_0$-subsequence of $w$). 

Assume that $u$ is a $\pmas$ of $w'$, for $p=p_0\sigma$. Assume also, for the sake of contradiction, that $u$ is a $p_0$-subsequence of $w$. Then, $u$ would also be a $p$-subsequence of $\phi(w)$, for $p=p_0\sigma$, a contradiction to the fact that $u$ is a $\pmas$ of the word $w'$. 

Now, assume that $u$ is not a $p_0$-subsequence of $w$. We want to show that $u$ is a $\pmas$ of $w'$. We can make several observations. Firstly, each subsequence of length $m-1$ of $u$ occurs in the word $u_1u_2\ldots u_{m-1}$, which contains as subsequences all words from $\Sigma^{m-1}$ (i.e., it is $(m-1)$-universal w.\,r.\,t. the alphabet $\Sigma$), and whose length is $(m-1)\sigma < p_0\sigma=p$. Thus, each subsequence of length $m-1$ of $u$ is a $p$-subsequence of $w'$. Secondly, $u$ is not a $p$-subsequence of $w'$. Indeed, if $u$ would be a $p$-subsequence of $w'$, then (due to the large factor $\$^{3p_0\sigma}$, as well as due to the fact that $u$ does not contain $\$$) we have that $u$ must be either a $p$-subsequence of $\phi(w)$ or a $p$-subsequence of $u_1\ldots u_m$. On the one hand, $u$ is not a subsequence of $u_1u_2\ldots u_{m-1}$ by the construction of $u_1u_2\ldots u_{m-1}$. Indeed, in order to match $u$ inside $u_1u_2\ldots u_{m-1}$ , the only possibility is to match the $i^{th}$ letter of $u$ to the last letter of $u_i$. As $u[m-1]$ is matched to the last letter of $u_{m-1}$, there is no position of $u_1u_2\ldots u_{m-1}$ to which $u[m]$ could be matched. So $u$ cannot a $p$-subsequence of $u_1u_2\ldots u_{m-1}$. On the other hand, if $u$ would be a $p$-subsequence of $\phi(w)$, then, clearly, $u$ would be a $p_0$-subsequence of $w$. Indeed, none of the letters of $u$ are $\$$-symbols, so they should be matched to the letters from $\Sigma$, which occur in a range of size $p_0\sigma$ in $\phi(w)$. As there are at most $p_0$ such letters, our claim holds. But $u$ is not (by hypothesis) a $p_0$-subsequence of $w$. So, putting all together $u$ is indeed a $\pmas$ of $w'$.

This shows the correctness of our reduction from $\pSubSeqMatch$ to $\pMAS$. Moreover, the time needed to construct the instance of $\pMAS$ is linear in the size of the instance of $\pSubSeqMatch$ (because $\sigma\in \mathcal{O}(1)$), so the conclusion follows.
\end{proof}

\section{Application: Subsequence Matching in Circular Words}

An interesting application of the framework developed in the previous setting is related to the notion of circular words.

We start with a series of preliminary definitions and results.

Intuitively, a circular word $w_\circ$ is a word obtained from a (linear) word $w\in\Sigma^*$ by linking its first symbol after its last one as shown in \cref{fig:circ_word}. 
\begin{figure}[h]
\vspace{-0.5cm}
	\centering
	\includegraphics[scale=0.5]{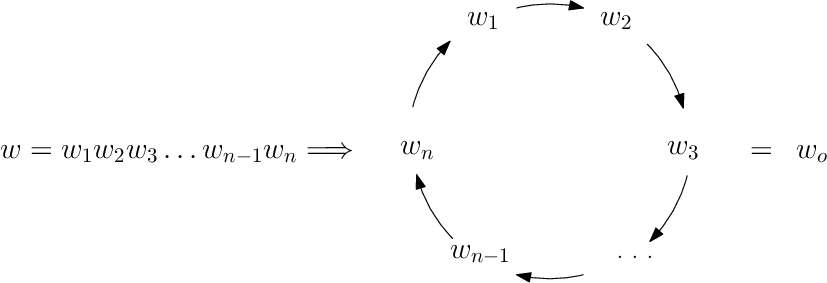}
	\caption{The circular word $w_\circ$, defined via the (linear) word $w$.}
	\label{fig:circ_word}
\vspace{-0.5cm}
\end{figure}

Formally, two words $u,w\in \Sigma^*$ are conjugates (denoted $u\sim w$) if there exist $x,y\in \Sigma^*$ such that $u=xy$ and $w=yx$. The conjugacy relation $\sim$ is an equivalence relation on $\Sigma^*$, and the circular word $w_\circ$ is defined as the equivalency class of $w$ with respect to $\sim$ (i.\,e., the set of all words equivalent to $w$ w.\,r.\,t. $\sim$). Clearly, for a word $w$ of length $n$, the equivalence class of $w$ with respect to $\sim$ has at most $n$ elements. It is worth investigating how we can represent circular words. To this end, we use the following definition from \cite{NagyProperties}.

\begin{definition}\label{def:circ_word_rep}
A pair $(u,n)\in\Sigma^* \times \mathbb{N}$ is a representation of the circular word $w_\circ$ if $|u| \leq n, n = |w|$ and $u^{\frac{n}{|u|}}\in w_\circ$. A minimal representation of a circular word $w_\circ$ is a pair $(u,n)$ such that $(u,n)$ is a representation of $w_\circ$ and for each other pair $(u',n)$ which represents $w_\circ$ we have that $|u|< |u'|$ or $|u|=|u'|$ and $u$ is lexicographically smaller than $u'$. 
\end{definition}

As an example, $(baa, 5)$ is a representation of the circular word $baaba_\circ$, because $(baa)^{5/3}=baaba$, but its minimal representation is $(ab,5)$. Indeed, $ab^{5/2}=ababa$, which is a conjugate of $baaba$, as $baaba=ba\cdot aba$ and $ababa=aba\cdot ba$.

Clearly, every circular word $w_\circ$ has a minimal representation, 
and an open problem from \cite{NagyProperties} is how efficiently can the minimal representation of a circular word $w_\circ$ be computed. We solve this open problem.\looseness=-1
\begin{theorem}\label[Theorem]{theorem:rep_linear}
Given a word $w$ of length $n$, the minimal representative $(v,n)$ of $w_\circ$ and a conjugate $w[i:n]w[1:i-1]=v^{n/|v|}$ of $w$ can be computed in $\mathcal{O}(n)$ time.\looseness=-1
\end{theorem}
\begin{proof}
We begin by referring to \cite{DumitranGM17}. In Lemma 5 of the respective paper it is shown that, for a word $u$ of length $n$, we can compute in $\mathcal{O}(n)$ time the values $SC[i] = \max\{|r| \mid r$ is both a suffix of $w[1:i-1]$ and a prefix of $w[i:n]\}$. The proof of that lemma can be directly adapted to prove the following result: given a word $u$ of length $n$ and an integer $\Delta \leq n$, we can compute in $\mathcal{O}(n)$ time the values $SC[i] = \max\{|r| \mid r$ is both a suffix of $w[1:i-1]$ and a prefix of $w[i:n], |r|\leq \Delta\}$. We will use this in the following.

Let us now prove the statement of our theorem. We consider the word $\alpha=www$, and we compute the array $SC[i] = \max\{|r| \mid r$ is both a suffix of $\alpha[1:i-1]$ and a prefix of $\alpha[i:3n], |r|\leq n-1\}$, using the result mentioned above. Now, we note that, for $i\in [n+1:2n]$ if $SC[i]=k$, then $\alpha[i:i+k-1]$ is the longest non-trivial border of the conjugate $w[i-n:n]w[1:i-n-1]=\alpha[i:i+n-1]$ of $w$; that means that  $\alpha[i:i+k-1]$ is the longest string which is both non-trivial suffix and prefix of $w[i-n:n]w[1:i-n-1]$. Consequently, the length of the shortest period of $w[i-n:n]w[1:i-n-1]$ is $n-k$. 

In conclusion, we have computed for each conjugate $w[j:n]w[1:j-1]$ of $w$ its shortest period $n-SC[j+n]$. Further, we can sort these conjugates w.\,r.\,t. their shortest period using counting sort. In this way, we obtain a conjugate $w[j:n]w[1:j-1]$ of $w$ which has the shortest period among all the conjugates of $w$. In the case of multiple such conjugates $w[j_g:n]w[1:j_g-1]$, with $g\in [1:\ell]$ for some $\ell$, we proceed as follows. We construct (in linear time) the suffix array of $\alpha$ (as in \cite{KarkkainenSB06}), and set $j= j_h$ where $j_h+n$ occurs as the first in the suffix array of $\alpha$ among all positions $j_g+n$, for $g\in [1:\ell]$. Therefore, we obtain a conjugate $w'=w[j:n]w[1:j-1]$ of $w$ which has the shortest period among all the conjugates of $w$, and is lexicographically smaller than all other conjugates of $w$ which have the same period. Moreover, for $p=n-SC[j+n]$, we have that $w'=(w'[1:p])^{n/p}$. \looseness=-1

The algorithm computing $w'$ and its period runs, clearly, in $\mathcal{O}(n)$ time, as all its steps can be implemented in linear time. Our statement is, thus, correct.
\end{proof}

This concludes the preliminaries part of this section.

In this framework, we define subsequences of circular words. 
\begin{definition}\label{def:circularSubseq}
A word $v$ is a subsequence of a circular word $w_\circ$ ($v \leq_\circ w$) if and only if there exists a conjugate $w' = w[i + 1:n] w[1:i]$ of $w$ such that $v$ is a subsequence of $w'$.
\end{definition}
This definition follows from \cite{Barker2020} where one defines $k$-circular universal words as words $u\in \Sigma^*$ which have at least one conjugate $u'$ whose set of subsequences of length $k$ is $\Sigma^k$. In this setting, we can define the following problem.
\begin{problemdescription}\label{p:ssqmatchCirc}
  \problemtitle{Circular Subsequence Matching, $\SubSeqMatch_\circ$}
  \probleminput{Two words $u$ and $w$ over $\Sigma$, with $\len{u} = m$, $\len{w} = n$, and $m\leq n$.}
  \problemquestion{Decide whether $v \leq_\circ w$.}
\end{problemdescription}

As the conjugates of a word $w$, of length $n$, are the factors of length $n$ of $ww$,
we immediately obtain the following result from \cref{thm:algoMatching}.
\begin{theorem}
 $\SubSeqMatch_\circ$ can be solved in $\mathcal{O}(mn)$ time.
\end{theorem}

And furthermore, the proof of \cref{thm:lowerBoundSubseq} shows that the following statement also holds.
\begin{theorem}\label[Theorem]{thm:lowerBoundSubseqCirc}
$\SubSeqMatch_\circ$ cannot be solved in time $\mathcal{O}(n^h m^g)$, where $h+g= 2-\epsilon$ with $\epsilon>0$, unless $\OVH$ fails. 
\end{theorem}

To conclude this paper, we consider an extension of the $\SubSeqMatch_\circ$ which seems natural to us. We begin by noting that reading (or, more precisely, traversing all the positions of) a circular word $w_\circ$ can be interpreted as reading (respectively, going through) the letters written around a circle, as drawn, for instance, in \cref{fig:circ_word}. So, we can start reading the word at some point on this circle, then go once around the circle, until we are back at the starting point. Then, as in a loop, we could repeat reading (traversing) the word. So, it seems natural to ask how many times do we need to read/traverse a circular word $w_\circ$ until we have that a given word $u$ is a subsequence of the word we have read/traversed.

Clearly, this problem is not well defined, as it depends on the starting point from which we start reading the circular word $w_\circ$. Let us consider an example.  Consider the word $w=ababcc$. Now, if we consider the circular word $w_\circ$, and we start reading/traversing it from position one of $w$ (i.\,e., we start reading $ababcc\cdot ababcc\cdot ababcc\cdot \ldots$) then we need to read/traverse twice the circular word to have that $ca$ is a subsequence of the traversed word. But if we start reading/traversing the circular word on any position $i\geq 2$ of $w$ (e.\,g., we start on position $3$ and read  $abccab\cdot abccab\cdot abccab\cdot \ldots$), then it is enough to traverse the circular word once to have that $ca$ is a subsequence of the traversed word. 

In this setting, there are two natural ways to fix this issue. 

The first one is to define a canonical point of start for the traversal. A natural choice for this starting point is to consider a special position of the word such as, for instance, the positions where a minimal representative of $w_\circ$ occurs. To this end, a conjugate $u=w[i:n]w[1:i-1]$ of a word $w$ of length $n$ is called minimal rotation of $w$ if $u=v^{n/|v|}$ and $(v,n)$ is a minimal representative of $w_\circ$. We obtain the following problem (presented here as a decision problem).
\begin{problemdescription}\label{p:ssqmatchIterated}
  \problemtitle{Iterated Circular Subsequence Matching, $\ItSubSeqMatch_\circ$}
  \probleminput{An integer $\ell$, a word $v$, and a word $w$, which defines the circular word $w_\circ$, over $\Sigma$, with $\len{v} = m$ and $\len{w} = n$, and $m\leq n$.}
  \problemquestion{Decide whether $v \leq u^\ell$, where $u$ is a minimal rotation of $w$.}
\end{problemdescription}
This problem is clearly well defined now, as if $u$ and $u'$ are minimal rotations of $w$, then $u=u'$. Moreover, this problem can be also formulated as a minimisation problem by simply asking for the smallest $\ell$ for which $\ItSubSeqMatch_\circ$ with input $(\ell,v,w)$ can be answered positively. 

The second way to solve the issue identified above is as follows. 
\begin{problemdescription}\label{p:ssqmatchIteratedMin}
  \problemtitle{Best Iterated Circular Subsequence Matching, $\BestItSubSeqMatch_\circ$}
  \probleminput{An integer $\ell$, a word $v$, and a word $w$, which defines the circular word $w_\circ$, over $\Sigma$, with $\len{v} = m$ and $\len{w} = n$, and $m\leq n$.}
  \problemquestion{Decide whether there is a conjugate $u$ of $w$ such that $v \leq u^\ell$.}
\end{problemdescription}
Clearly, this problem can be also formulated as a minimisation problem by simply asking for the smallest $\ell$ for which $\BestItSubSeqMatch_\circ$ with input $(\ell,v,w)$ can be answered positively. 

Our results regarding these two problems are summarized below.
\begin{theorem}\label[Theorem]{thm:ItMatch}
\begin{enumerate}
\item The problem $\ItSubSeqMatch_\circ$ (and the related minimisation problem) can be solved in $\mathcal{O}(\min(n\sigma+m, n + m\log\log n))$ time, where $\sigma=|\Sigma|$.
\item $\BestItSubSeqMatch_\circ$ (and the corresponding minimisation problem) can be solved in $\mathcal{O}(nm)$ time.
\item $\BestItSubSeqMatch_\circ$ cannot be solved in time $\mathcal{O}(n^h m^g)$, where $h+g= 2-\epsilon$ with $\epsilon>0$, unless $\OVH$ fails. 
\end{enumerate}
\end{theorem}
\begin{proof}
Let us first observe that we can assume without loss of generality that $\sigma\leq m+1$. Indeed, if $w$ contains letters which do not occur in $v$, we can simply replace in $w$ those letters by a fresh letter $\$ \notin \al(w)$ and obtain a new word $w' \in (\al(v)\cup\{\$\})^*$. Then, each problem above can be answered positively for the input $(\ell,v,w)$ if and only if the same problem can be answered positively for the input $(\ell,v,w')$. Moreover, we can assume that all the letters of $\Sigma$ occur in $w$; otherwise, we can simply remove from $\Sigma$ the letters that do not occur in $w$ (and answer each of the problem negatively if $v$ contains such a letter). 

Let $w'=ww$. In a preprocessing phase which we use in the algorithms proving statements 1 and 2, we will construct the $n\times \sigma$ matrix $\nextpos[\cdot][\cdot]$, where $\nextpos[i][a] = (\min\{j\mid i<j\leq 2n, w'[j]=a, |w[i+1:j-1]|_a=0\}) \mod n$, for $i\leq n$ and $a\in \Sigma$. Intuitively, $\nextpos[i][a]$ is the first position of the circular word $w_\circ$ where $a$ occurs after position $i$, where the numbering of the positions of $w_\circ$ is defined to coincide with the numbering of the positions of $w$. 

We claim that the elements of the matrix $\nextpos[\cdot][\cdot]$ can be computed in $\mathcal{O}(n\sigma)$ time. For each letter $a\in \Sigma$, we can compute in linear time $\mathcal{O}(n)$ the increasingly ordered list $L_a=j_1<\ldots <j_k$ of all the positions of $w$ where $a$ occurs. Then, for $i\in [1:j_1-1]$ we set $\nextpos[i][a]= j_1$. Further, for $i$ from $2$ to $k-1$, we set $\nextpos[i][a]=j_i$ for $i\in [j_{i-1},j_i-1]$. Finally, for $i\in [j_k:n]$ we set $\nextpos[i][a]= j_1$. Clearly, this computation is correct and can be completed in the stated time.

Now, to solve $ \ItSubSeqMatch_\circ$, we first compute $u=w[i:n]w[1:i-1]$, the minimal rotation of $w$, using the algorithm from \cref{theorem:rep_linear}. Then, we set a variable $count=0$ and a variable $temp=\nextpos[i-1][v[1]]$, and run the following loop, for $j$ from $2$ to $m$:
\begin{enumerate}
\item $temp' \gets \nextpos[temp][v[i]]$
\item if $temp \geq temp' \geq i$ then $count \gets count +1$ 
\item $temp \gets temp'$
\end{enumerate} 
We return $count+1$ as the smallest $\ell$ for which $\ItSubSeqMatch_\circ$ with input $(\ell,v,w)$ can be answered positively.

This process clearly runs in $\mathcal{O}(m)$ time, and we claim that it is correct. Indeed, it simply identifies the letters of $v$, in a greedy fashion, in $u^\omega $. As such our algorithm identifies the shortest prefix of $u^\omega $ which contains $v$ as a subsequence (and this corresponds to the smallest $\ell$ for which $\ItSubSeqMatch_\circ$ with input $(\ell,v,w)$ can be answered positively). Moreover, we use the variable $count$ to count how many times we completely went around the circular words (i.\,e., $count$ is increased every time we pass position $i$ of $w$ in our search for the next letter of $v$). Therefore, it is correct to return $count+1$ as the smallest $\ell$ for which $\ItSubSeqMatch_\circ$ with input $(\ell,v,w)$ can be answered positively. 

So, the algorithm above solves $\ItSubSeqMatch_\circ$ in $\mathcal{O}(n\sigma + m)$ time (including the preprocessing). 

In the case when we want to skip the $\mathcal{O}(n\sigma)$ preprocessing step (where the matrix $\nextpos[\cdot][\cdot]$ is computed), we can proceed as follows. For each list $L_a$, we compute in $\mathcal{O}(|L_a|)$ time a data structure allowing us to do successor searches in $L_a$ in time $\mathcal{O}(\log\log n)$ per query; see \cite{FisGaw}[Proposition 7 in the arXiv version]. Overall, this preprocessing takes $\mathcal{O}(n)$ time. Now, whenever we need to compute $\nextpos[i][a]$ in step 1 of the algorithm solving $\ItSubSeqMatch_\circ$ described above, we proceed as follows. We compute the successor of $i$ in $L_a$. If this successor is defined (i.\,e., there is some $j>i$ in $L_a$), then we return it as the value of $\nextpos[i][a]$. If the successor is not defined (i.\,e., $i$ is greater than all elements of $L_a$), we return the smallest element of $L_a$ as the value of $\nextpos[i][a]$. With this change, the algorithm described above solves $\ItSubSeqMatch_\circ$ in $\mathcal{O}(n + m\log \log n)$ time (including the preprocessing). This completes the proof of the first statement.

We now show the second statement. The approach simply uses the (first) algorithm described above for $\ItSubSeqMatch_\circ$ for each conjugate $u=w[i:n]w[1:i-1]$ of $w$ instead of the minimal rotation. We then simply return the conjugate $u$ of $w$ for which the value $\ell$ computed by the respective algorithm is the smallest. This is clearly correct (as it is an exhaustive search, in the end). The complexity of this procedure is $\mathcal{O}(n\sigma + mn)$ (as the preprocessing is only done once, then we run the loop which runs in $\mathcal{O}(m)$ time from the solution of $\ItSubSeqMatch_\circ$ for each conjugate of $w$, and there are $n$ such conjugates). As $\sigma\leq m+1$, we have that our solution to $\BestItSubSeqMatch_\circ$ runs in $\mathcal{O}(mn)$ time.

To show the third statement (and, as such, the optimality of the algorithm presented above, conditional to $\OVH$), it is enough to note that $\SubSeqMatch_\circ$ is the same as $\BestItSubSeqMatch_\circ$ for the input integer $\ell=1$. Therefore, this statement holds due to \cref{thm:lowerBoundSubseqCirc}.
\end{proof}

As a comment on the previous result, it remains open whether $\ItSubSeqMatch_\circ$ can be solved more efficiently than with our algorithms, whose complexity is given in statement 1 of \cref{thm:ItMatch}. 
\section{Conclusions}
In this paper we have considered a series of classical matching and analysis problems related to the occurrences of subsequences in words, and extended them to the case of subsequences occurring in bounded ranges in words. In general, we have shown that the matching problem, where we simply check if a word is a subsequence of another word, becomes computationally harder in this extended setting: it now requires rectangular time. Similarly, problems like checking whether two words have the same set of subsequences of given length or checking whether a word contains all possible words of given length as subsequences become much harder (i.\,e., NP-hard as opposed to solvable in linear time) when we consider subsequences in bounded ranges instead of arbitrary subsequences. We have also analysed a series of problems related to absent subsequences in bounded ranges of words and, again, we have seen that this case is fundamentally different than the case of arbitrary subsequences. In general, our results paint a comprehensive picture of the complexity of matching and analysis problems for subsequences in bounded ranges.

As an application of our results, we have considered the matching problem for subsequences in circular words, where we simply check if a word $u$ is a subsequence of any conjugate of another word $w$ (i.\,e., is $v$ a subsequence of the circular word $w_\circ$), and we have shown that this problem requires quadratic time. A series of other results regarding the occurrences of subsequences in circular words were discussed, but there are also a few interesting questions which remained open in this setting: What is the complexity of deciding whether two circular words have the same set of subsequences of given length? What is the complexity of checking whether a circular word contains all possible words of given length as subsequences? Note that the techniques we have used to show hardness in the the case of analysis problems for subsequences in bounded ranges of words do not seem to work in the case of circular words, so new approaches would be needed in this case.

\newpage
\bibliographystyle{fundam}
\bibliography{references}

\newpage



\end{document}